\documentclass[pra,twocolumn,preprintnumbers,nofootinbib,amsmath,amssymb,superscriptaddress,showpacs,aps,10pt]{revtex4-1}

\usepackage{graphicx, color, graphpap} 
\usepackage{amsmath}
\usepackage{enumitem}
\usepackage{amssymb}
\usepackage{amsthm}
\usepackage{pstricks}
\usepackage{multirow}
\usepackage{hyperref}
\usepackage[T1]{fontenc}

\long\def\ca#1\cb{} 

\newcommand{\ket}[1]{|#1\rangle}               
\newcommand{\proj}[1]{|#1\rangle\langle #1|}   

\newcommand{\bra}[1]{\langle #1|}              
\newcommand{\dya}[1]{\ket{#1}\!\bra{#1}}
\newcommand{\dyad}[2]{\ket{#1}\!\bra{#2}}      
\newcommand{\ip}[2]{\langle #1|#2\rangle}      
\newcommand{\guess}{\text{guess}}

\newcommand{\rhot}{\tilde{\rho}}

\newcommand{\FC}{\mathcal{F}}

\newcommand{\HC}{\mathcal{H}}
\newcommand{\IC}{\mathcal{I}}

\newcommand{\Tr}{{\rm Tr}}

\newcommand{\pg}{\text{pg}}

\renewcommand{\geq}{\geqslant}
\renewcommand{\leq}{\leqslant}

\newcommand{\ot}{\otimes}
\newcommand{\ad}{^\dagger}

\newcommand*{\id}{\openone}

\newcommand*{\hmin}{H_{\min}}

\newcommand{\ep}{\epsilon}

\renewcommand{\th}{\theta } 

\newcommand{\Th}{\Theta }

\newcommand{\Lm}{\Lambda }

\newcommand{\sg}{\sigma }

\newcommand{\om}{\omega }

\newtheoremstyle{example}{\topsep}{\topsep}%
{} 
{} 
{\bfseries} 
{.} 
{   } 
{\thmname{#1}\thmnumber{ #2}} 
\theoremstyle{example}

\newtheorem{theorem}{Theorem}
\newtheorem{lemma}[theorem]{Lemma}
\newtheorem{corollary}[theorem]{Corollary}

\theoremstyle{definition}

\begin{document}

\title{An equality between entanglement and uncertainty}

\author{Mario Berta}
\affiliation{Institute for Quantum Information and Matter, Caltech, Pasadena, CA 91125, USA}
\author{Patrick J. Coles}
\affiliation{Institute for Quantum Computing and Department of Physics and Astronomy, University of Waterloo, N2L3G1 Waterloo, Ontario, Canada}
\affiliation{Centre for Quantum Technologies, National University of Singapore, 2 Science Drive 3, 117543 Singapore.}
\author{Stephanie Wehner}
\affiliation{Centre for Quantum Technologies, National University of Singapore, 2 Science Drive 3, 117543 Singapore.}
\affiliation{QuTech, Delft University of Technology, Lorentzweg 1, 2628 CJ Delft, Netherlands}
\date{\today}

\begin{abstract}
Heisenberg's uncertainty principle implies that if one party (Alice) prepares a system and randomly measures one of two incompatible observables, then another party (Bob) cannot perfectly predict the measurement outcomes. This implication assumes that Bob does not possess an additional system that is entangled to the measured one; indeed the seminal paper of Einstein, Podolsky and Rosen (EPR) showed that maximal entanglement allows Bob to perfectly win this guessing game. Although not in contradiction, the observations made by EPR and Heisenberg illustrate two extreme cases of the interplay between entanglement and uncertainty. On the one hand, no entanglement means that Bob's predictions must display some uncertainty. Yet on the other hand, maximal entanglement means that there is no more uncertainty at all. Here we follow an operational approach and give an exact relation - an equality - between the amount of uncertainty as measured by the guessing probability, and the amount of entanglement as measured by the recoverable entanglement fidelity. From this equality we deduce a simple criterion for witnessing bipartite entanglement and a novel entanglement monogamy equality.
\end{abstract}

\pacs{03.67.-a, 03.67.Hk}

\maketitle


\section{Uncertainty relations}

Heisenberg's uncertainty principle forms one of the fundamental elements of quantum mechanics. Originally proven for measurements of position and momentum, it is one of the most striking examples of the difference between a quantum and a classical world~\cite{Heisenberg}. Uncertainty relations today are probably best known in the form given by Robertson~\cite{Robertson}, who extended Heisenberg's result to two arbitrary observables $X$ and $Z$. More precisely, Robertson's relation states that when measuring the state $\ket{\psi}$ using either $X$ or $Z$, then
\begin{align}\label{eq:heisenberg}
\Delta X \Delta Z \geq \frac{1}{2} |\bra{\psi}[X,Z]\ket{\psi}|\ ,
\end{align}
where $\Delta Y = \sqrt{\bra{\psi}Y^2\ket{\psi} - \bra{\psi}Y\ket{\psi}^2}$ for $Y \in \{X,Z\}$ is the standard deviation resulting from measuring $\ket{\psi}$ with observable $Y$.
 
In the modern day literature, uncertainty is usually measured in terms of entropies (starting with~\cite{Hirschman57,Beckner75,BiaMyc75}, see~\cite{EURreview1} for a survey). One of the reasons this is desirable is that~\eqref{eq:heisenberg} makes no statement if $\ket{\psi}$ happens to give zero expectation on $[X,Z]$~\cite{deutsch}. To see how uncertainty can be quantified in terms of 
entropies, let us start with a simple example. Throughout, we let Alice ($A$) denote the system to be measured. For now, let us consider measuring a single qubit in the state $\rho_A$ using two incompatible measurements given by the Pauli $\sigma_x$ or $\sigma_z$ eigenbases, and let $K$ be the random variable associated with the measurement outcome. We have from~\cite{MaaUff88} that for any state $\rho_A$
\begin{align}\label{eq:basicUR}
H(K|\Theta) = \frac{1}{2}\Big[H(K|\Theta=\sigma_x)+H(K|\Theta=\sigma_z)\Big]\geq\frac{1}{2}\ ,
\end{align}
where $H(K|\Theta=\theta) = - \sum_{k} p_{k|\Theta=\theta} \log p_{k|\Theta=\theta}$ is the Shannon entropy (all logarithms are base 2 in this article) of the probability distribution over measurement outcomes $k \in \{0,1\}$ when we perform the measurement labeled $\theta$ on the state $\rho_A$, and each measurement is chosen with probability $p_\theta = 1/2$. To see that this is an uncertainty relation note that if one of the two entropies is zero, then~\eqref{eq:basicUR} tells us that the other is necessarily non-zero, i.e., there is at least some amount of uncertainty. If we measure a $d_A$-dimensional system $A$ in two orthonormal bases ${\theta}_0 = \{\ket{x_0}\}_{x=1}^{d_A}$ and ${\theta_1} = \{\ket{x_1}\}_{x=1}^{d_A}$ then the r.h.s. of~\eqref{eq:basicUR} becomes $\log(1/c)$, where $c = \max_{x_0,x_1} |\ip{x_0}{x_1}|^2$. The largest amount of uncertainty, i.e., the largest $\log(1/c)$, is thereby obtained when $|\ip{x_0}{x_1}| = 1/\sqrt{d_A}$, that is, the two bases are mutually unbiased (MUB)~\cite{Ivanovic1981}.

When thinking about uncertainty, it is often illustrative to adopt an adversarial perspective and consider an ``uncertainty game''~\cite{BertaEtAl}, commonly used in quantum cryptography~\cite{bounded}. In particular, we will think about uncertainty from the perspective of an observer called Bob holding a second system ($B$) whose task is to guess the outcome of the measurement on Alice's system successfully. Bob thereby knows ahead of time what measurements could be made and the probability that a particular measurement setting is chosen. To help him win the game, Bob may even prepare $\rho_A$ himself, and Alice tells him which measurement she performed before he has to make his guess. The amount of uncertainty as measured by entropies can be understood as a limit on how well Bob can guess Alice's measurement outcome - the more difficult it is for Bob to guess the more uncertain Alice's measurement outcomes are. If Bob is not entangled with $A$ but only keeps classical information about the state, such as for example a description of the density operator $\rho_A$, then~\eqref{eq:basicUR} still holds even if we condition on Bob's classical information $B$~\cite{Hall1}. More precisely, we have $H(K|\Theta B_{\rm classical}) \geq 1/2$ for any states or distribution of states that Bob may prepare.


\section{Uncertainty and entanglement}


Another central element of quantum mechanics is the possibility of entanglement, and examples suggest that there is a strong interplay between entanglement and uncertainty. In particular, Einstein, Poldolsky and Rosen~\cite{Einstein} observed, that 
if Bob is maximally entangled with $A$ then his uncertainty can be reduced dramatically. To see this imagine that $\rho_{AB} =\proj{\Phi}$ where $\ket{\Phi}=(\ket{00}+\ket{11})/\sqrt{2}$ is the maximally entangled state between $A$ and $B$. Since $\ket{\Phi}$ is maximally correlated in both the $\sg_x$ and $\sg_z$ eigenbases, Bob can simply measure his half of the EPR pair in the same basis as Alice to predict her measurement outcome perfectly, winning the guessing game described above. This is precisely the effect observed in~\cite{Einstein} and highlights that the uncertainty relations of~\eqref{eq:heisenberg} and~\eqref{eq:basicUR} do not capture the interplay between entanglement and uncertainty in the general two-party guessing game. Fortunately, it is possible to extend the notion of uncertainty relations to take the possibility of entanglement into account~\cite{RenesBoileau}. Such relations are known as uncertainty relations with quantum side information (here $B$). More precisely, it was shown~\cite{BertaEtAl} that if we measure $A$ in two bases labeled $\theta_0,\theta_1$ then
\begin{align}\label{eqn2}
H(K|B\Theta)&=\frac{1}{2}\Big[H(K|B\Theta=\theta_0)+H(K|B\Theta=\theta_1)\Big]\notag\\
&\geq\log(1/c)+H(A|B)\ ,
\end{align}
where $H(A|B)$ is the conditional von Neumann entropy of $A$ given $B$. If $A$ and $B$ are entangled, then $H(A|B)$ can be negative. Indeed, $H(A|B)=-\log d_A$ when $\rho_{AB}$ is the maximally entangled state, in which case the lower bound in~\eqref{eqn2} becomes trivial. The uncertainty relation of~\eqref{eqn2} thus allows for the possibility that Bob's uncertainty could be reduced in the presence of entanglement. It also provides us with a first clue to the relation between entanglement and uncertainty in one direction, namely that little uncertainty (i.e., $H(K|B\Theta)$ is small) implies that $H(A|B)$ must be negative and hence $\rho_{AB}$ is entangled~\cite{DevWin05}. As such,~\eqref{eqn2} is useful for the task of witnessing entanglement~\cite{PHCFR,LXXLG}.

Many more similar relations have since been proven for more than two measurements on Alice's system, and in terms of other forms of entropies. One entropy measure that is of central importance in cryptography is the conditional min-entropy $\hmin$, and it yields a more immediate link between uncertainty relations with quantum side information and the uncertainty game mentioned above. Specifically, it was shown~\cite{6670761} that if we measure $A$ in one of $d_A+1$ possible mutually unbiased bases chosen uniformly at random then 
\begin{align}\label{eq:hminPrevious}
\hmin(K|B\Theta)\gtrsim\log d_A+\min\{0,\hmin(A|B)\}\ .
\end{align}
(More precisely, smoothing of the entropies is required for~\eqref{eq:hminPrevious} to hold, and hence the symbol $\gtrsim$ refers to an additional term that depends on the smoothing.) With $K$ being a classical random variable, the conditional min-entropy $\hmin(K|B\Theta)=-\log P_{\rm guess}(K|B\Theta)$ is simply derived from the maximum probability that Bob can guess $K$, averaged over the choice of basis $\theta$~\cite{KonRenSch09}. That is, it captures exactly how well Bob can guess Alice's measurement outcome $K$ by performing a measurement on $B$. The fully quantum conditional min-entropy $\hmin(A|B)$ has the operational interpretation $\hmin(A|B)=-\log [d_A\cdot F(A|B)]$, with
\begin{align}
F(A|B)=\max_{\Lambda_{B\rightarrow A'}}F\left(\Phi_{AA'},\mathcal{I}_A\otimes\Lambda_{B\rightarrow A'}(\rho_{AB})\right)\ ,
\end{align}
where $F(\rho, \sg) = \left(\Tr \sqrt{\sqrt{\rho}\sg\sqrt{\rho}}\right)^2$ is Uhlmann's fidelity~\cite{Uhlmann76}, $\Phi_{AA'} = \dya{\Phi_{AA'}}$, and $\ket{\Phi_{AA'}}=(1/\sqrt{d_A})\cdot\sum_{j=1}^{d_A}\ket{j}_A\ket{j}_{A'}$ is the maximally entangled state between $A$ and $A'$~\cite{KonRenSch09}. In other words, the conditional min-entropy $\hmin(A|B)$ measures how close one can bring a bipartite quantum state $\rho_{AB}$ to the maximally entangled state by performing an arbitrary operation $\Lm$ on the $B$ system. Recall from the example above that Bob can win the uncertainty game perfectly if $\rho_{AB}$ really is maximally entangled. Intuitively, the conditional min-entropy thus measures how far away Bob is from this scenario. Needless to say one could write down similar statements for R{\'e}nyi entropies other than the min-entropy, but these are in fact equivalent up to small error terms.

Do these relations resolve the question of how uncertainty relates to entanglement? Note that the uncertainty relation~\eqref{eq:hminPrevious} again provides us with a relation between entanglement and uncertainty in one direction. In particular, it tells us that if it is easy for Bob to guess Alice's measurement outcome ($\hmin(K|B\Theta)$ is small), then there really exists some map $\Lambda_{B\rightarrow A'}$ that Bob can use to bring $\rho_{AB}$ at least somewhat close to being maximally entangled with $A$. That is, it tells us that a reduction in uncertainty implies the presence of entanglement. However, it does not tell us that the presence of entanglement really does lead to a significant reduction in uncertainty. Of course, if $\rho_{AB}$ is close to the maximally entangled state then uncertainty is reduced by at least some amount, because two states which are close yield similar statistics when measured. Yet we will see below that this alone is insufficient for our purpose.


\section{Main result}\label{sec:main}

Here, we prove the following finite-dimensional equality if we measure $A$ in one of $d_A +1$ possible mutually unbiased bases with uniformly random probability
\begin{align}\label{eq:mainResult}
H_2(K|B\Theta)=\log(d_A+1)-\log\left(2^{-H_2(A|B)}+1\right)\ , 
\end{align}
where
\begin{align}
&H_2(A|B)=\notag\\
&-\log\Tr\left[\rho_{AB}(\id_A\otimes \rho_B)^{-1/2}\rho_{AB}(\id_A\otimes \rho_B)^{-1/2}\right]
\end{align}
is the conditional R{\'e}nyi $2$-entropy used in quantum cryptography (see e.g.~\cite{RennerThesis05URL,Muller-Lennert:2013aa}), and $K$ is the classical measurement outcome obtained by measuring $A$ in the basis labeled $\Theta=\theta$. Since $K$ is a classical random variable, the R{\'e}nyi $2$-entropy $H_2(K|B\Theta)$ has an operational interpretation as given by the probability that Bob manages to guess Alice's measurement outcome $K$ using the pretty good measurement~\cite{Hausladen94,PhysRevA.78.022316} after he learns which measurement $\Theta$ was made: $H_2(K|B\Theta)=-\log P_{\rm guess}^{\pg}(K|B\Theta)$. For the fully quantum R{\'e}nyi $2$-entropy $H_2(A|B)$ we prove (see Appendix~\ref{sec:methods}) that
\begin{eqnarray}\label{eq:h2qoperational}
H_2(A|B)&=-\log[d_A\cdot F^{\pg}(A|B)], \quad\text{with }\\ 
F^{\pg}(A|B)&=F(\Phi_{AA'},\mathcal{I}_A\otimes\Lambda_{B\rightarrow A'}^{\pg}(\rho_{AB}))\ ,  
\end{eqnarray}
where $\Lambda^{\pg}$ is the pretty good recovery map~\cite{BarnKnil02}. (We denote by $\Phi_{AA'}$ the normalized maximally entangled state, and hence a factor $d_{A}$ appears in~\eqref{eq:h2qoperational}.) Both the pretty good measurement and the pretty good recovery map get very close to the performance of the optimal processes~\cite{Hausladen94,BarnKnil02}. For the special case $\rho_{AB} = \rho_A \otimes \rho_B$, \eqref{eq:mainResult} becomes an equation relating unconditional entropies that was discussed in~\cite{Ivanovic1992,BrukZeilPRL1999}, and used to derive uncertainty relations for Shannon entropies (see~\cite{EURreview1} for an overview). Furthermore, the conditional R{\'e}nyi $2$-entropy also appears in the study of randomness extractors against quantum side information (see e.g.~\cite{RennerThesis05URL,TRSS10}), and in its quantum counterpart decoupling (see e.g.~\cite{Dupuis10,Dup09}).

We mention that the existence of a full set of MUBs is only known in prime power dimension~\cite{Bandy01,Wootters89}, but we show in Appendix~\ref{app:revisited} that our main result~\eqref{eq:mainResult} also holds for informationally complete positive operator valued measures (SIC-POVMs) and unitary 2-designs. For the latter efficient constructions are known in any dimension~\cite{Gross07,Dankert09}, and in particular the set of all bases defines a unitary 2-design.

Our relation~\eqref{eq:mainResult} establishes an equivalence between uncertainty as measured by $H_2(K|B\Theta)$ and our ability to recover entanglement as given by $H_2(A|B)$. It is an operational way to merge the observations of EPR and Heisenberg into a single equation, demonstrating that both effects can be seen as flip sides of the same coin.


\section{Discussion}\label{sec:disc}

\subsection{Operational examples} 

To gain further intuition about~\eqref{eq:mainResult}, let us first return to the uncertainty game discussed earlier. Note that in terms of the operational interpretations of the conditional R\'{e}nyi $2$-entropy, we can rewrite~\eqref{eq:mainResult} as 
\begin{align}\label{eq:mainOperational}
P_{\rm guess}^{\pg}(K|B\Theta)&= \frac{1}{d_A+1}\sum_{\theta} P_{\rm guess}^{\pg}(K|B\Theta=\theta)\notag\\
&= \frac{d_A\cdot F^{\pg}(A|B) + 1}{d_A+1}\ .
\end{align}
In the game, Bob prepares a state $\rho_{AB}$ and sends the $A$ system to Alice. She measures $A$ in one basis chosen uniformly at random from the complete set of $d_A+1$ MUBs, and announces the basis (the index $\theta$) to Bob. Bob's task is to guess Alice's outcome using the pretty good measurement on $B$. Equation~\eqref{eq:mainOperational} says that Bob's ability to win or lose this game is quantitatively connected to the recoverable entanglement fidelity of $\rho_{AB}$, as measured by $F^{\pg}(A|B)$.  

Let us consider a number of special cases that illustrate this concept. In what follows we refer to $H_2(A|B)\geq 0$ as the Heisenberg-limited regime and $H_2(A|B)< 0$ as the enhanced regime. As we will see below, this terminology refers to two distinct regimes, one in which Bob's guessing probability is restricted by a Heisenberg-like uncertainty relation (see Eq.~\eqref{eq:HeisenbergLimited1}), and the other in which his guessing probability can be enhanced beyond this restriction (although of course not in violation of the uncertainty principle). For example, if $\rho_{AB}$ is the maximally entangled state, we have $F^{\pg}(A|B)=1$ and Bob can guess Alice's measurement outcome perfectly regardless of which measurement she performs, i.e., $P_{\rm guess}^{\pg}(K|B\Theta=\theta)=1$ for all $\theta$. That is, there is no uncertainty as expected. If Bob prepares $\rho_{AB}$ with less than maximal entanglement, then $F^{\pg}(A|B)<1$ and there will be at least one basis for which Bob cannot perfectly guess the outcome. Thus, there is at least some amount of uncertainty expressed quantitatively as $H_2(K|B\Theta)$. If $\rho_{AB}$ is separable, then Bob is stuck in the Heisenberg-limited regime ($F^{\pg} \leq 1/d_A$) and his ability to guess is very poor, constrained by the uncertainty relation
\begin{align}\label{eq:HeisenbergLimited1}
P_{\rm guess}^{\pg}(K|B\Theta) \leq 2/(d_A+1).
\end{align}
This illustrates that entanglement is necessary for Bob to gain an advantage in the guessing game.


\subsection{Uncertainty and certainty relations}

One might ask why we formulate our uncertainty equality~\eqref{eq:mainResult} using a full set of $d_A+1$ MUBs, can we not use fewer measurements? To answer this, it is instructive to study what kind of relations our main result~\eqref{eq:mainResult} implies. On the one hand, we can deduce regular uncertainty relations, and, e.g., we get a relation in terms of the smooth conditional min-entropy similar to~\cite{6670761},
\begin{align}\label{eqn662}
H_{\min}^{\varepsilon}(K| \Th B)\geq\log(d_{A}+1)&-\log\left(2^{-H_{\min}(A|B)}+1\right)\notag\\
&-1-2\log\frac{1}{\ep}\ ,
\end{align}
where $\ep>0$ denotes a small error term (see Appendix~\ref{app:crypto} for details). Here, the l.h.s.~has the operation meaning of minus the logarithm of Bob's guessing probability (up to error $\ep$) when Alice measures in one of $d_A+1$ possible MUBs chosen uniformly at random. But on the other hand, we also get relations that upper bound the uncertainties of incompatible observables. In the literature these are known as certainty relations~\cite{SanchezRuiz1995,Wehner09}, and here we give the first such relations that allow for quantum side information. E.g.~we get in terms of the conditional min-entropy (again up to a small error term $\ep>0$),
\begin{align}\label{eqn663}
H_{\min}(K| \Th B)\leq\log(d_A+1)&-\log\left(2^{-H_{\min}^{\ep}(A|B)}+\frac{2}{\ep^{2}}\right)\notag\\
&+1+2\log\frac{1}{\ep}\ .
\end{align}
This says that Bob's certainty, i.e., his ability to guess Alice's measurement outcome, must be high if he is highly entangled to Alice as measured by the smooth conditional min-entropy. Like our main result, \eqref{eqn663} implies that if Alice and Bob are maximally entangled, Bob has perfect certainty about Alice's outcomes regardless of which measurement she performs.

Now there is a simple argument that considering less than a complete set of MUBs, a so-called extendable set (where there exist an MUB that could be added to the set), implies that only trivial certainty relations can hold. As uncertainty equalities as in~\eqref{eq:mainResult} imply non-trivial certainty relations, such equalities cannot hold for extendable sets. This is in sharp contrast to uncertainty relations, where non-trivial relations can be obtained for just two measurements. To see this, consider the case where $\rho_A$ is just one qubit and we perform measurements in the $\sg_X$ and $\sg_Z$ eigenbases, respectively. Hence $B$ is trivial, $\hmin(K|\Theta B) = \hmin(K|\Theta)$, and we just consider the entropy of the outcome distribution of measuring $\rho_A$ in one of the two bases. In terms of the uncertainty game discussed before this means that Bob can only choose the state $\rho_{A}$ in order to guess the outcome of the measurement on Alice's system (but is not allowed to keep any quantum information $B$ about $A$). Clearly, when $\rho_A$ is an eigenstate of $\sg_Y$, then the outcome distribution for both $\sg_X$ and $\sg_Z$ is uniform and hence $\hmin(K|\Theta) = 1$, which is the maximum value. This argument generalizes to any extendable set of MUBs, since there exists a state (from another MUB) that has $\hmin(K|\Theta) = \log d_A$ which is the maximum value that it can take and hence only the trivial upper bound/certainty relation holds. It is thus clear that equalities such as~\eqref{eq:mainResult} can only hold for sets of measurements that are sufficiently rich.


\subsection{Bounds for fewer bases}

Even though there does not exist an uncertainty equality for measuring in less than $d_{A}+1$ MUBs, we can still give lower (and trivial upper) bounds for Bob's uncertainty about $1\leq n \leq d_A $ MUBs on $A$ in terms of the recoverable entanglement fidelity between $A$ and $B$. Moreover these inequalities are tight for all $n$, that is, fixing the set of measurements, there exists states that achieve the upper and lower bounds. Our relations are again in terms of the conditional R\'enyi 2-entropy. Using $P^{\pg}_{\text{guess}}(n)$ as a shorthand to denote Bob's guessing probability $P^{\pg}_{\text{guess}}(K | B\Th)$ when Alice does measurements in a subset of size $n$ of a complete set of MUBs, we find that the following bounds are tight in the Heisenberg-limited regime ($F^{\pg}(A|B)\leq 1/d_A$),
\begin{align}\label{eqn8}
\frac{1}{d_A} \leq P^{\pg}_{\text{guess}}(n) \leq \frac{d_A}{n}\cdot F^{\pg}(A|B)+\frac{n-1}{n\cdot d_A}\ .
\end{align}
Moreover, the following bounds are tight in the enhanced regime ($F^{\pg}(A|B)> 1/d_A$),
\begin{align}\label{eqn89}
F^{\pg}(A|B) \leq P^{\pg}_{\text{guess}}(n) \leq \frac{n-1}{n}\cdot F^{\pg}(A|B)+\frac{1}{n}\ .
\end{align}
Here, the upper bounds (lower bounds) can be thought of as uncertainty relations (certainty relations). Note that we can derive these tight relations for all $n$ directly from our main result. Taken together, \eqref{eqn8} and~\eqref{eqn89} completely characterize the allowable range of values that $P^{\pg}_{\text{guess}}(n)$ can attain. As an example, consider $d_A=5$. Fig.~\ref{fgr1} plots the tight upper and lower bounds as a function of $F^{\pg}(A|B)$. Notice that the (trivial) lower bound does not change as $n$ varies from 1 to 5; it only increases when we include the sixth basis. This is consistent with our discussion in the previous subsection, where we noted that a non-trivial certainty relation, i.e., a relation stronger than the certainty relation one obtains for a single basis, requires sufficiently rich sets of measurements. In contrast, the tight upper bound steadily decreases with $n$, reflecting the complementarity between the different bases. Overall, as $n$ increases from 1 to 6, the area of the allowed range monotonically shrinks towards zero, and the allowed area becomes zero for $n=6$ since the two quantities $P^{\pg}_{\text{guess}}(n)$ and $F^{\pg}(A|B)$ are deterministically related by our main result.

\begin{figure}
\includegraphics[width=\linewidth]{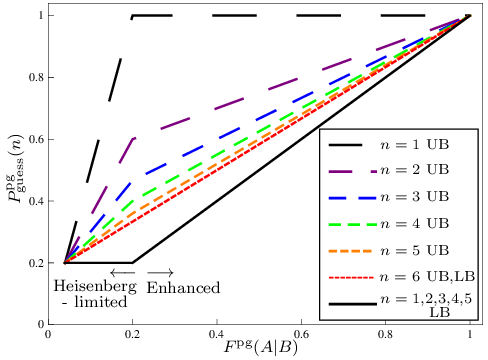}
\caption{Upper (UB) and lower (LB) bounds on $P^{\pg}_{\text{guess}}(n)$ as a function of $F^{\pg}(A|B)$ for $d_A=5$, for various values of $n$. For $n=1,2,3,4,5$, the (trivial) lower bound is given by the solid black line. The upper bounds for these values of $n$ are respectively the black, purple, blue, green, and orange dashed lines. For $n=6$, the upper and lower bounds coincide, i.e., the allowed values are confined to live on the red dashed line. The regions where $H_2(A|B)\geq 0$ and $H_2(A|B)< 0$ are labeled ``Heisenberg-limited'' and ``enhanced'', respectively.
\label{fgr1}}
\end{figure}

From Fig.~\ref{fgr1} one can also see that if Bob can guess two MUBs on $A$ well then he can also guess $d_A+1$ MUBs on $A$ fairly well. Conceptually this follows from a two step chain of reasoning: if Bob's uncertainty is low for two MUBs, then he must be entangled to Alice, which in turn implies that he must have a low uncertainty for all bases. So entanglement provides the key link, from two MUBs to all bases. From the above results, it is straightforward to derive the following quantitative statement of this idea
\begin{align}\label{eqn10}
P^{\pg}_{\text{guess}}(d_A+1) \geq \frac{d_A\cdot\left(2P^{\pg}_{\text{guess}}(2)-1\right) +1}{d_A+1}\ ,
\end{align}
which says that as $P^{\pg}_{\text{guess}}(2)\to 1$, then $P^{\pg}_{\text{guess}}(d_A+1)\to 1$.


\subsection{Applications in quantum information theory}

We briefly discuss some applications of our main result to quantum information processing tasks. Because entanglement is crucial for several quantum information technologies, the experimenter often needs a method to verify that their source is indeed producing entangled pairs, i.e., an ``entanglement witness''. Following~\cite{PhysRevA.68.032103,BertaEtAl,PhysRevA.67.022320,PhysRevA.70.022316}, our main result offers a simple strategy for witnessing entanglement since it connects entanglement to uncertainty, which is experimentally measurable. In particular, Alice and Bob (in their distant labs, receiving $A$ and $B$ respectively) can sample from the source multiple times and communicate their results to gather statistics, say, regarding the $K_{\th}$ observable on $A$ and the $L_{\th}$ observable on $B$. Suppose they do this for a set of $n$ MUBs $\{K_{\th}\}_{\th=1}^n$ on $A$, with Bob measuring in a some arbitrary set of $n$ bases $\{L_{\th}\} _{\th=1}^n$ on $B$. They then estimate the joint probability distribution for each pair $\{K_{\th},L_{\th}\}$, and hence they can evaluate the classical entropies $H_2(K_{\th}|L_{\th})$. According to our main result, their source is necessarily entangled if
\begin{align}\label{eq:entwitness}
\sum_{\th=1}^n 2^{-H_2(K_{\th}|L_{\th})}>1+\frac{n-1}{d_A}\ .
\end{align}
Note that this method offers the flexibility of witnessing entanglement with $2\leq n\leq d_A +1$ observables (see also~\cite{PhysRevA.86.022311} for a different approach). For $n=2$, the same strategy based on the uncertainty relation~\eqref{eqn2} was implemented in~\cite{PHCFR, LXXLG}.

Another application of our main result is to quantum error correction of noisy quantum channels. By viewing~\eqref{eq:mainResult} from the dynamic perspective of Alice sending states through a quantum channel to Bob, and noting that the entanglement fidelity such as that appearing in~\eqref{eq:h2qoperational} is a standard figure of merit for quantum error correction, we see from~\eqref{eq:mainResult} that Bob's ability to error-correct is quantitatively linked to his ability to guess which states Alice sends, when she is sending basis elements from a complete set of MUBs.  

Our uncertainty equality~\eqref{eq:mainResult} also gives insights for studying the monogamy of correlations. The basic idea of monogamy is that $A$'s entanglement with $B$ limits the degree to which $A$ can be entangled with a third system, $E$. There have been several statements of monogamy in the literature; however, a nice aspect of our results is the potential to state monogamy as an equation rather than an inequality. We show that for any tripartite pure state $\rho_{ABE}$,
\begin{align}\label{eqn10001}
&\inf_{\sigma}D_{\frac{1}{2}}\left(\rho_{AE}\|\frac{\id_{A}}{d_{A}}\otimes\sigma_{E}\right)\notag\\
&=\log d_A -\log\Big[(d_{A}+1)\cdot P^{\pg}_{\text{guess}}(K|B \Theta)
 -1\Big]\ ,
\end{align} 
where the infimum is over all quantum states $\sigma_{E}$, and $D_{\frac{1}{2}}$ denotes the relative R\'enyi $1/2$-entropy (see Appendix~\ref{app:monogamy} for details). This relation states that Bob's guessing probability for a complete set of Alice's MUBs is a quantitative measure of the distance of $\rho_{AE}$ (Alice's and Eve's state) to a completely uncorrelated state. According to~\eqref{eqn10001}, Bob and Eve fight in a ``zero-sum game'' to be correlated to Alice, i.e., any gain of knowledge about Alice's system by Bob forces Eve's state to get closer to being uncorrelated with Alice, and conversely any gain of distance from the uncorrelated state by Eve forces Bob to lose knowledge.

Finally, we remark that the conditional R\'enyi 2-entropy is an important quantity in the study of classical and quantum randomness extractors against quantum side information (see e.g.~\cite{TRSS10,Dup09}). Since our uncertainty equality~\eqref{eq:mainResult} connects the conditional R\'enyi 2-entropy of the pre-measurement state to the conditional R\'enyi 2-entropy of the post-measurement state, our main result~\eqref{eq:mainResult} shines some light on the relation between classical and quantum extractors. It can be used to get a new perspective on the results in~\cite{6670761}, where security of the noisy storage model~\cite{Wullschleger09} was first linked to the quantum capacity.


\subsection{Conclusions}

In summary, we considered a two-party guessing game, where Alice measures her system $A$ in one of $n$ possible complementary observables and Bob uses his system $B$ to help him guess Alice's outcome. We showed that Bob's probability for winning this game, assuming he does the ``pretty good measurement'' on $B$, is connected through an equality for $n = d_A+1$ (inequality for $n < d_A+1$) to the prior entanglement between $B$ and $A$. The latter is measured by the entanglement fidelity that can be recovered with the ``pretty good recovery map'', which we proved is given by the conditional R\'enyi 2-entropy. We therefore showed that our operationally-motivated equality can be thought of as an entropic uncertainty relation, and as such, connects Heisenberg's uncertainty principle to EPR's guessing game via an equation. We expect our approach to inspire further quantitative relations capturing the connection between uncertainty and entanglement. In addition, it would be interesting to explore the connection between the guessing game considered here and other non-local uncertainty games that have been considered in the literature, e.g., in the context of Bell inequalities \cite{RevModPhys.86.419} and steering \cite{PhysRevLett.98.140402,Smith:2012aa}.


\section*{Acknowledgments}

PJC thanks J\k{e}drzej Kaniewski for helpful discussions. PJC and SW acknowledge support from the National Research Foundation and Ministry of Education, Singapore.


\appendix

\section{Quantitative measures}\label{sec:methods}

\subsection{Entanglement}

Despite not being monotonic under local operations and classical communication, conditional entropies play an important role in entanglement theory~\cite{HHHH09}. For example, the conditional von Neumann entropy $H(A|B)$ quantifies the asymptotic rate for distilling EPR pairs via a one-way hashing protocol~\cite{DevWin05}. Another important conditional entropy studied in cryptography is the min-entropy. It was originally defined in an abstract form~\cite{RennerThesis05URL}, but was later given an intuitive operational meaning~\cite{KonRenSch09} in terms of the recoverable entanglement fidelity, i.e., $H_{\min}(A|B)=-\log [d_A\cdot F(A|B)]$ with
\begin{align}\label{eq:minop}
F(A|B)=\max_{\Lambda_{B\rightarrow A'}}F\left(\Phi_{AA'}, (\IC_{A} \ot \Lambda_{B\rightarrow A'})(\rho_{AB})\right)\ ,
\end{align}
where the maximum is over all quantum operations $\Lambda_{B\rightarrow A'}$ with $A'$ a copy of $A$. (See, e.g.~\cite{NieChu00} for discussion of the importance of the entanglement fidelity in quantum information theory.)

A related entropy measure is the conditional R\'enyi 2-entropy, which is defined as
\begin{align}\label{eq:h2def}
&H_2(A|B)=\notag\\
&-\log \Tr\left[\rho_{AB}(\id_A\otimes \rho_B)^{-1/2}\rho_{AB}(\id_A\otimes \rho_B)^{-1/2}\right]\ .
\end{align}
Here we give an operational meaning for the conditional R\'enyi 2-entropy by showing that, like the conditional min-entropy, it is linked to the recoverable entanglement fidelity in that $H_2(A|B)=-\log [d_A\cdot F^{\pg}(A|B)]$ with
\begin{align}\label{eq:h2alt}
F^{\pg}(A|B)=F(\Phi_{AA'},\mathcal{I}_{A}\otimes \Lambda_{B\rightarrow A'}^{\pg}(\rho_{AB}))\ ,
\end{align}
and $\Lambda^{\pg}_{B\rightarrow A'}$ is the pretty good recovery map. To see this, we note that the pretty good recovery map can be written as
\begin{align}
\Lambda^{\pg}_{B\rightarrow A'}(\cdot)=\frac{1}{d_A}\cdot\mathcal{E}\ad_{B\rightarrow A'}\left(\rho_B^{-1/2}(\cdot) \rho_B^{-1/2}\right)\ ,
\end{align}
where $\mathcal{E}\ad_{B\rightarrow A'}$ denotes the adjoint of the Choi-Jamilkowski map of $\rho_{AB}$,
\begin{align}
\mathcal{E}_{A\rightarrow B}(\cdot)=d_A\cdot\Tr_{A}\left[\left((\cdot)^T \ot \id_{B}\right)\rho_{AB}\right]\ .
\end{align}
Putting this in~\eqref{eq:h2def} we arrive at~\eqref{eq:h2alt}. The map $\Lambda^{\pg}_{B\rightarrow A'}$ is pretty good in the sense that it is close to optimal for recovering the maximally entangled state, i.e., we have~\cite{BarnKnil02}
\begin{align}\label{eq:FpgClose}
F^{2}(A|B)\leq F^{\pg}(A|B)\leq F(A|B)\ .
\end{align}
We also remark that both $F(A|B)$ and $F^{\pg}(A|B)$ are non-increasing under the action of local quantum channels acting on system $B$ and local unital (identity-preserving) quantum channels acting on system $A$ (see, e.g.~\cite{Muller-Lennert:2013aa}).


\subsection{Uncertainty}

When measuring a bipartite quantum state $\rho_{AB}$ on $A$ in some basis $K=\{\ket{k}\}$, we arrive at a classical-quantum state
\begin{align}\label{eq:cqstate}
\rho_{KB}&=\sum_k (\dya{k}\ot\id_{B})\rho_{AB}(\dya{k}\ot\id_{B})\\
&=\sum_{k}\dya{k}\otimes\rho_{B}^{k}\ .
\end{align}
The conditional min-entropy of $\rho_{KB}$, using the formula for $F(K|B)$ from~\eqref{eq:minop}, translates to $H_{\min}(K|B)=-\log P_{\text{guess}}(K|B)$ with
\begin{align}
P_{\text{guess}}(K|B)=\max_{\{E^{k}_{B}\}}\sum_{k}\Tr\left[E_{B}^{k}\rho_{B}^{k}\right]\ ,
\end{align}
the probability for guessing $K$ correctly by performing the optimal measurement $\{E^{k}_{B}\}$ on the quantum side information $B$. The conditional min-entropy quantifies the uncertainty of $K$ in the exact sense of the uncertainty game, namely it quantifies the probability that Bob wins the uncertainty game.

The conditional R\'enyi 2-entropy of a classical-quantum state is again defined as in~\eqref{eq:h2def}. Furthermore, it was shown in~\cite{PhysRevA.78.022316} that its operational form~\eqref{eq:h2alt} is given by
\begin{align}\label{eq:h2altcq}
H_{2}(K|B)=-\log P^{\pg}_{\text{guess}}(K|B)\ ,
\end{align}
where $P^{\pg}_{\text{guess}}(K|B)$ denotes the probability of guessing $K$ by performing the ``pretty good measurement''~\cite{Hausladen94}. For the classical-quantum state~\eqref{eq:cqstate} the pretty good measurement operators are defined as
\begin{align}
\Pi^{k}_{B}=\rho_{B}^{-1/2}\rho_{B}^{k}\rho_{B}^{-1/2}\ .
\end{align}
By calculating $P^{\pg}_{\text{guess}}(K|B)=\sum_{k}\Tr\left[\Pi_{B}^{k}\rho_{B}^{k}\right]$, the equivalence of~\eqref{eq:h2altcq} to the definition of the conditional R\'enyi 2-entropy in~\eqref{eq:h2def} can be seen. Hence, the conditional R\'enyi 2-entropy corresponds to the probability that Bob wins the uncertainty game by using the pretty good measurement. It is known that the pretty good measurement performs close to optimal, i.e., analogous to~\eqref{eq:FpgClose} we have~\cite{Hausladen94}
\begin{align}\label{eq:PpgClose}
P^{2}_{\text{guess}}(K|B)\leq P^{\pg}_{\text{guess}}(K|B)\leq P_{\text{guess}}(K|B)\ .
\end{align}

In the following we will not only measure in one fixed basis, but with equal probability in one of $d_{A}+1$ MUBs. For that reason, we will work with the state
\begin{align}
&\rho_{KB\Theta}=\frac{1}{d_{A}+1}\notag\\
&\sum_{\theta=1}^{d_{A}+1}\sum_{k=1}^{d_{A}}(\dya{\theta_{k}}\ot\id_{B})\rho_{AB}(\dya{\theta_{k}}\ot\id_{B})\otimes\dya{\theta}_{\Theta}\ ,
\end{align}
where the elements of the $d_{A}+1$ MUBs $\theta$ are denoted by $\{\ket{\theta_{k}}\}$. It is straightforward to see that
\begin{align}
P^{\pg}_{\text{guess}}(K|B\Theta)=\frac{1}{d_{A}+1}\cdot\sum_{\theta}P^{\pg}_{\text{guess}}(K|B\Theta=\theta)\ .
\end{align}


\section{Proof of main results}

\subsection{Full set of mutually unbiased bases}

Here we prove our main result, the uncertainty equality~\eqref{eq:mainResult}. For this we define $\rhot_{AB}= (\id_{A} \ot\rho_B^{-1/4})\rho_{AB}(\id_{A} \ot\rho_B^{-1/4})$ and rewrite the fully quantum conditional R\'enyi 2-entropy as $H_{2}(A|B)=-\log \Tr\left[\rhot_{AB}^2\right]$. Similarly, we rewrite the classical-quantum conditional R\'enyi 2-entropy as
\begin{align}
&H_{2}(K|B\Theta)=\notag\\
&-\log\left(\frac{1}{d_{A}+1}\cdot\sum_{\theta,k}\Tr_{B}\left[\Tr_{A}\left[\rhot_{AB}(\dya{\theta_k}\ot \id_{B})\right]^{2}\right]\right)\ .
\end{align}
Now we introduce the space $\HC_{A'B'}\cong \HC_{AB}$ as well as the state $\rhot_{A'B'} \cong \rhot_{AB}$. We have
\begin{align}
&(d_{A}+1)\cdot2^{-H_{2}(K|B\Theta)}\notag\\
&=\sum_{\theta,k}\Tr_B\big\{\Tr_A\big[(\dya{\theta_k}\ot \id_{B})\rhot_{AB}\big]\notag\\
&\qquad\qquad\quad\Tr_A\big[(\dya{\theta_k}\ot \id_{B}) \rhot_{AB}\big]\big\}\\
&=\sum_{\theta,k}\Tr_{BB'}\Tr_{AA'}\big[(\dya{\theta_k}\ot\dya{\theta_k})\notag\\
&\qquad\qquad\qquad\qquad\,\,\,(\rhot_{AB}\ot\rhot_{A'B'} ) F_{BB'}\big]\\
&=\Tr_{BB'}\Tr_{AA'}\big[(I_{AA'}+F_{AA'})(\rhot_{AB}\ot\rhot_{A'B'} ) F_{BB'}\big]\\
&=\Tr_{BB'}\Tr_{AA'}\big[(\rhot_{AB}\ot\rhot_{A'B'})F_{BB'}\big]\notag\\
&\quad+\Tr_{BB'}\Tr_{AA'}\big[F_{AA'}(\rhot_{AB}\ot\rhot_{A'B'})F_{BB'}\big]\\
&=\Tr_{B}\big[\Tr_A(\rhot_{AB} )\Tr_A( \rhot_{AB})\big]\notag\\
&\quad+\sum_{t,s}\Tr_{BB'}\Tr_{AA'}\big[(\dyad{t}{s}\ot \dyad{s}{t}\ot\id_{BB'})\notag\\
&\qquad\qquad\qquad\qquad\quad\,\,\,(\rhot_{AB}\ot\rhot_{A'B'})F_{BB'}\big]\\
&=1+\sum_{t,s}\Tr_{B}\big\{\Tr_A\big[(\dyad{t}{s}\ot \id)\rhot_{AB}\big]\notag\\
&\qquad\qquad\qquad\,\,\,\,\Tr_A\big[(\dyad{s}{t}\ot \id)\rhot_{AB}\big]\big\}\\
&=1+\Tr\left[\rhot_{AB}^2\right]\ ,\label{eq:mainproof}
\end{align}
where $F_{AA'}= \sum_{t,s}\dyad{t}{s}\ot \dyad{s}{t}$ is the operator that swaps $A$ and $A'$ (similarly for $F_{BB'}$). The second line uses the ``swap trick'', for operators $M$ and $N$, and swap operator $F$: $\Tr(MN)=\Tr[(M\ot N)F)$. The third line invokes that a full set of MUBs generates a complex projective 2-design~\cite{1523643}, that is,
\begin{align}
\sum_{\theta,k} \dya{\theta_{k}} \ot \dya{\theta_{k}}= I_{AA'}+F_{AA'}\ .
\end{align}
In Appendix~\ref{app:revisited}, we show that our result also holds for other measurements as long as they form a complex projective 2-design.


\subsection{Fewer bases}

Here, we derive the upper and lower bounds~\eqref{eqn8} and~\eqref{eqn89} on the uncertainty when Alice measures in $1\leq n<d_{A}+1$ MUBs. It is helpful to first analyze the case for one basis $K$.

\begin{lemma}
Let $K=\{\ket{k}\}$ be an orthonormal basis on some Hilbert space $\HC_A$. Then, we have for any bipartite quantum state $\rho_{AB}$ that
\begin{align}\label{eqn21}
P^{\pg}_{\guess}(K|B)\geq F^{\pg}(A|B)\ ,
\end{align}
where $\rho_{KB}=\sum_k (\dya{k}\ot \id_{B})\rho_{AB}(\dya{k}\ot \id_{B})$.
\end{lemma}

\begin{proof}
We calculate
\begin{align}\label{eqn22}
P^{\pg}_{\guess}(K|B)&= \Tr\left[\rho_{KB}\rho_B^{-1/2}\rho_{KB}\rho_B^{-1/2}\right]\\
&=\Tr\left[ \rho_{AB}\rho_B^{-1/2}\rho_{KB}\rho_B^{-1/2}\right]\\
&= d_A\cdot\Tr\left[\Phi_{AA'} (\IC_{A}\ot \Lambda^{\pg}_{B\rightarrow A'})(\rho_{KB})\right]\\
&=F\left(\Phi_{AA'}, d_A\cdot(\IC_{A}\ot\Lambda^{\pg}_{B\rightarrow A'})(\rho_{KB})\right)\\
&\geq F\left(\Phi_{AA'}, (\IC_{A}\ot \Lambda^{\pg}_{B\rightarrow A'})(\rho_{AB})\right)\\
&=F^{\pg}(A|B)\ .
\end{align}
The inequality step in this proof invoked the property that the fidelity decreases upon decreasing one of its arguments, and hence it remains to show
\begin{align}\label{eqn32049}
d_A\cdot(\IC_{A}\ot \Lambda^{\pg}_{B\rightarrow A'})(\rho_{KB})\geq (\IC_{A}\ot \Lambda^{\pg}_{B\rightarrow A'})(\rho_{AB})\ .
\end{align}
We denote the non-negative operator $\sg_{AA'}=(\IC_{A}\ot\Lambda^{\pg}_{B\rightarrow A'})(\rho_{AB} )$, and note that the measurement in $K$ on the $A$-system commutes with $\IC_{A}\ot\Lambda^{\pg}_{B\rightarrow A'}$. We get
\begin{align}
&d_A\cdot(\IC_{A}\ot \Lambda^{\pg}_{B\rightarrow A'})(\rho_{KB})-(\IC_{A}\ot \Lambda^{\pg}_{B\rightarrow A'})(\rho_{AB})\notag\\ 
&=d_A\cdot\sum_k (\dya{k}\ot \id_{A'})\sg_{AA'}(\dya{k}\ot \id_{A'})\notag\\
&\quad-\sum_{k,k'} (\dyad{k}{k}\ot \id_{A'}) \sg_{AA'} (\dyad{k'}{k'}\ot \id_{A'})\\
&=(d_A-1)[ \sum_k (\dya{k}\ot \id)\sg_{AA'}(\dya{k}\ot \id)\notag\\
&\quad-\frac{1}{d_A-1}\sum_{k,k'\neq k} (\dyad{k}{k}\ot \id) \sg_{AA'}(\dyad{k'}{k'}\ot \id)]\\
&=(d_A-1)(\FC\ot \IC)(\sg_{AA'})\ ,
\end{align}
where we set in the last line
\begin{align}
&\FC(\cdot)=\frac{1}{d_A-1}\sum_{m=1}^{d_A-1} Z^m (\cdot) ( Z^m)\ad,\quad Z = \sum_{k=0}^{d_A-1} \om^{k} \dya{k}\ ,\notag\\
&\om = e^{2\pi i /d_A}\ .
\end{align}
Since $\FC$ is a completely positive trace preserving map, the claim follows.
\end{proof}

Equation~\eqref{eqn21} states that the entanglement fidelity quantified by $F^{\pg}(A|B)$ lower bounds the guessing probability $P^{\pg}_{\text{guess}}(K|B)$. By combining~\eqref{eqn21} with our uncertainty equality~\eqref{eq:mainResult}, we get that for subset of size $1\leq n<d_A+1$ of a complete set of MUBs,
\begin{align}
&\sum_{\theta=1}^{n}P^{\pg}_{\text{guess}}(K|B\Theta=\theta)\notag\\
&=(n-1)\cdot F^{\pg}(A|B)+1+\Big((d_A+1-n)\cdot F^{\pg}(A|B)\notag\\
&\qquad\qquad\qquad\qquad\qquad\qquad-\sum_{\theta=n+1}^{d_A+1} P^{\pg}_{\text{guess}}(K|B\Theta=\theta)\Big)\\
&\leq (n-1)\cdot F^{\pg}(A|B)+1\ ,
\end{align}
and this proves~\eqref{eqn89}. Similarly, we can invoke the immediate relation $P^{\pg}_{\text{guess}}(K|B)\geq1/d_A$ to get
\begin{align}
&\sum_{\theta=1}^{n}P^{\pg}_{\text{guess}}(K|B\Theta=\theta)\notag\\
&=d_A\cdot F^{\pg}(A|B)+\frac{n-1}{d_A}\notag\\
&\quad+\left(\frac{d_A+1-n}{d_A}-\sum_{\theta=n+1}^{d_A+1} P^{\pg}_{\text{guess}}(K|B\Theta=\theta)\right)\\
&\leq d_A\cdot F^{\pg}(A|B)+\frac{n-1}{d_A}\ ,
\end{align}
and this proves~\eqref{eqn8}. The tightness of the bounds in~\eqref{eqn8} and~\eqref{eqn89} follows by construction. In the region $F^{\pg}(A|B)\geq 1/d_A$, the upper bound is achieved by a bipartite pure state whose Schmidt basis is one of the $\Theta$ bases appearing in the sum of guessing probabilities under consideration, and the lower bound is achieved by a bipartite pure state whose Schmidt basis is one of the $\Theta$ bases that belongs to the same complete MUB set as the bases under consideration, but whose guessing probability was removed from the sum under consideration.  In the region $F^{\pg}(A|B)\leq 1/d_A$, the upper bound is achieved by a tensor product state $\rho_A \ot \rho_B$ such that $\rho_A$ is diagonal in one of the $\Theta$ bases appearing in the sum of guessing probabilities under consideration, and the lower bound is similarly achieved by such a tensor product state where $\rho_A$ is diagonal in one of the $\Theta$ bases that belongs to the same complete MUB set as the bases under consideration, but whose guessing probability was removed from the sum under consideration.


\subsection{Monogamy of correlations}\label{app:monogamy}

Here we show~\eqref{eqn10001} from the main text. The precise statement is as follows.

\begin{corollary}
Let $\{\Theta\}_{\theta\in\Theta}$ be a complete set of MUBs on some Hilbert space $\HC_A$, and denote $\theta=\{\ket{\theta_{k}}\}_{k=1}^{d_{A}}$. Then, we have for any tripartite pure quantum state $\rho_{ABE}$ that
\begin{align}
&\inf_{\sigma}D_{\frac{1}{2}}\left(\rho_{AE}\|\frac{\id_{A}}{d_{A}}\otimes\sigma_{E}\right)=\notag\\
&\log d_A -\log\Big[(d_{A}+1)\cdot P^{\pg}_{\text{guess}}(K|B \Theta)
-1\Big]\ ,
\end{align} 
where the infimum is over all quantum states $\sigma_{E}$, and the relative R\'enyi $1/2$-entropy is given by
\begin{align}
&D_{\frac{1}{2}}\left(\rho_{AE}\|\frac{\id_{A}}{d_{A}}\otimes\sigma_{E}\right)=\notag\\
&-\log\left(\mathrm{tr}\left[\rho_{AE}^{1/2}\left(\frac{\id_{A}}{d_{A}}\otimes\sigma_{E}\right)^{1/2}\right]\right)^{2}\ .
\end{align}
\end{corollary}

\begin{proof}
For any conditional entropy that is invariant under local isometries on the conditioning system, one can define a dual entropy. For some generic entropy $H_K$, the dual entropy $H^{\mathrm{dual}}_{K}$ is defined by
\begin{align}
H_K(A|B)=-H^{\mathrm{dual}}_{K}(A|E)_{\rho}\ ,
\end{align}
where $E$ is a system that purifies $\rho_{AB}$. Since $H_2(A|B)$ is invariant under local isometries on $B$, the dual entropy is well defined, and it is known that~\cite{Tomamichel13}
\begin{align}\label{eq:h2dual}
-H^{\mathrm{dual}}_{2}(A|E)=\inf_{\sigma}D_{\frac{1}{2}}(\rho_{AE}\|\id_{A}\otimes\sigma_{E})\ .
\end{align}
By the standard rewriting
\begin{align}\label{eq:0conrel}
D_{\frac{1}{2}}\left(\rho_{AE}\|\id_{A}\otimes\sigma_{E}\right)=D_{\frac{1}{2}}\left(\rho_{AE}||\frac{\id_{A}}{d_A}\ot \sigma_{E}\right)-\log d_A \ ,
\end{align}
the claim follows from our main result~\eqref{eq:mainResult}.
\end{proof}


\section{Applications}

\subsection{Witnessing entanglement}

Here we show the origin of~\eqref{eq:entwitness}, our condition for witnessing entanglement. We will make use of the following lemma, which says that separable states cannot have a negative conditional entropy.

\begin{lemma}\label{lem:seperable}
Let $\rho_{AB}$ be a separable quantum state. Then, we have that
\begin{align}
H_2(A|B)\geq H_{\min}(A|B) \geq 0\ .
\end{align}
\end{lemma}

\begin{proof}
The inequality $H_2(A|B)\geq H_{\min}(A|B)$ holds for any quantum state $\rho_{AB}$ since $\Lambda^{\pg}_{B\rightarrow A'}$ in~\eqref{eq:h2alt} is a particular map, and the conditional min-entropy involves an optimization over all maps $\Lambda_{B\rightarrow A'}$ in~\eqref{eq:minop}.

To prove $H_{\min}(A|B)\geq 0$, note that any local operation on a separable state results in another separable state. Now suppose $\sg_{AA'}=(\IC_A \ot \hat{\Lambda}_{B\rightarrow A'}) (\rho_{AB})$ is the separable state that achieves the optimization when evaluating the conditional min-entropy for $\rho_{AB}$ (i.e., $\hat{\Lambda}$ is the optimal channel in~\eqref{eq:h2alt}. Then, we have that
\begin{align}
F(A|B)=F\left(\Phi_{AA'}, \sg_{AA'}\right)&\leq F\left(\Phi_{AA'}, \id_A \ot \sg_{A'}\right)\\
&=1/d_A\ ,
\end{align}
which follows because the fidelity increases upon increasing one of its arguments, and because for separable $\sg_{AA'}$ we have $\sg_{AA'}\leq \id_A \ot \sg_{A'}$ with $\sg_{A'}=\Tr_A(\sg_{AA'})$. Using $H_{\min}(A|B)=-\log [d_A\cdot F(A|B) ]$, the claim follows.
\end{proof}

The following is our criterion for witnessing entanglement.

\begin{lemma}
Let $\rho_{AB}$ be a separable quantum state. Let $\{K_{\th}\}_{\th=1}^n$ be a subset (of size $n$) of a complete set of MUBs on $A$, and let $\{L_{\th}\} _{\th=1}^n$ be an arbitrary set of $n$ orthonormal bases on $B$. Then, we have that
\begin{align}\label{eq:sepcondition}
\sum_{\th=1}^n 2^{-H_2(K_{\th}|L_{\th})}\leq1+\frac{n-1}{d_A}\ ,
\end{align}
where
\begin{align}
\rho_{K_{\th}L_{\th}}=\sum_{p,q}&(\dya{K_{\th,p}}\ot \dya{L_{\th,q}})\rho_{AB}\notag\\
&(\dya{K_{\th,p}}\ot \dya{L_{\th,q}})\ .
\end{align}
\end{lemma}

\begin{proof}
Since $\rho_{AB}$ is separable, the previous lemma (Lemma~\ref{lem:seperable}) tells us that $F^{\pg}(A|B)\leq 1/d_A$. Combining this with the bound in~\eqref{eqn8} we have
\begin{align}
\frac{n-1}{d_A}+1&\geq(n-1)\cdot F^{\pg}(A|B)+1\\
&\geq\sum_{\th=1}^n P^{\pg}_{\text{guess}}(K_{\th}|B)\\
&=\sum_{\th=1}^n 2^{-H_2(K_{\th}|B)}\\
&\geq\sum_{\th=1}^n 2^{-H_2(K_{\th}|L_{\th})}\ ,
\end{align}
where the last inequality follows because the conditional R{\'e}nyi $2$-entropy satisfies the data-processing inequality~\cite{Muller-Lennert:2013aa}.
\end{proof}


\subsection{Quantum cryptography}\label{app:crypto}

Here we show the uncertainty and certainty relations~\eqref{eqn662} and~\eqref{eqn663} in terms of the smooth conditional min-entropy. For a bipartite quantum state $\rho_{AB}$ and smoothing parameter $\varepsilon\geq0$, the smooth conditional min-entropy is defined as
\begin{align}
H^{\varepsilon}_{\min}(A|B)_{\rho}=\sup_{\bar{\rho}_{AB}}H_{\min}(A|B)_{\bar{\rho}}\ ,
\end{align}
where the supremum is over all sub-normalized states $\bar{\rho}_{AB}$ on $AB$ that are $\varepsilon$-close to $\rho_{AB}$ in purified distance~\cite{TomamichelThesis2012}. Now, the crucial point is that $H^{\varepsilon}_{\min}$ and $H_2$ are equivalent in the following sense: it holds for any bipartite quantum state $\rho_{AB}$ and $\varepsilon>0$ that (see e.g.~\cite[Lemma A.25]{Berta13}),
\begin{align}\label{eq:H2Hmin}
H_{\min}(A|B)\leq H_{2}(A|B)\leq H^{\varepsilon}_{\min}(A|B)+\log\frac{2}{\varepsilon^{2}}\ .
\end{align}
By combining~\eqref{eq:H2Hmin} with our main result~\eqref{eq:mainResult}, we obtain the uncertainty and certainty relation for the smooth conditional min-entropy as given in~\eqref{eqn662} and~\eqref{eqn663}. We note that similar uncertainty relations have been derived in~\cite{6670761,Dupuis:2013aa}, and were used to analyze security in the noisy storage model.


\section{Extensions of main result}

\subsection{Complex projective 2-designs}\label{app:revisited}

We have seen in Appendix~\ref{sec:methods} that the proof of our uncertainty equality~\eqref{eq:mainResult} crucially relies on the fact that a full set of MUBs generates a complex projective 2-design. In general, a complex projective 2-design is a set $\{\ket{\psi_y}\}_{y\in Y}$ (of size $|Y|$) of vectors $\ket{\psi_y}$ lying in a Hilbert space $\mathcal{H}_{A}$ such that
\begin{align}\label{eq:twodesignidentity}
\frac{1}{|Y|}\cdot\sum_{y\in Y}\dya{\psi_y}^{\otimes2}=\frac{1}{d_{A}\cdot(d_{A}+1)}(\id_{AA'}+F_{AA'})\ ,
\end{align}
where $F_{AA'}$ denotes the swap operator, and $A'$ a copy of $A$. It turns out that there are other ``informationally equivalent'' measurements that generate a complex projective 2-design. As an example we mention SIC-POVMs, such as the four states forming a tetrahedron on the Bloch sphere for qubits.

\begin{corollary}
Let $ \{\frac{1}{d_A}\cdot\dya{\psi_k}\}_{k=1}^{d_A^2}$ be a SIC-POVM on some Hilbert space $\mathcal{H}_{A}$. Then, we have for any bipartite quantum state $\rho_{AB}$ that
\begin{align}\label{eqn26}
 H_2(K|B)=\log[d_A(d_{A}+1)]-\log\left(2^{-H_2(A|B)}+1\right)\ ,
\end{align}
where
\begin{align}
\rho_{KB}= \sum_{k=1}^{d_{A}^{2}} \dya{k}\ot \Tr_A[(\frac{1}{d_A}\cdot\dya{\psi_k} \ot \id_{B})\rho_{AB}]
\end{align}
is a classical-quantum state with $\{\ket{k}\}$ an orthonormal basis on $\HC_K$.
\end{corollary}

Notice that the $d_A$-dependent term on the r.h.s.\ of~\eqref{eqn26} is slightly different from the corresponding term appearing in our main result~\eqref{eq:mainResult}, and indeed~\eqref{eqn26} implies that $\log d_A \leq H_2(K|B)\leq 2\log d_A$ for SIC-POVMs. Nevertheless, the proof of~\eqref{eqn26} is identical to the proof of~\eqref{eq:mainResult}, with the appropriate version of~\eqref{eq:twodesignidentity} substituted into the proof.

In addition, so-called unitary 2-designs are closely related to complex projective 2-designs. A set $\{U_{y}\}_{y\in Y}$ (of size $|Y|$) of unitaries $U_{y}$ on some Hilbert space $\mathcal{H}_{A}$ forms a unitary two-design if
\begin{align}
\frac{1}{|Y|}\cdot\sum_{y\in Y}\left(U_{y}\otimes U_{y}^{\dagger}\right)^{\otimes2}=\int_{U(d_{A})} \left(U\otimes U^{\dagger}\right)^{\otimes2}dU\ ,
\end{align}
where the integration  is over all unitaries with respect to the Haar measure. Examples of unitary two-designs are the full unitary group, or the Clifford group for qubit systems. In fact, unitary 2-designs generate complex projective 2-designs.

\begin{lemma}
Let $\{U_{\theta}\}_{\theta\in\Theta}$ be a unitary two-design on some Hilbert space $\mathcal{H}_{A}$. Then, we have
\begin{align}
&\frac{1}{d_{A}\cdot|\Theta|}\cdot\sum_{k=1}^{d_{A}}\sum_{\theta\in\Theta}(U_{\theta}\ket{k}\bra{k}U_{\theta}^{\dagger})^{\otimes2}\notag\\
&=\frac{1}{d_{A}\cdot(d_{A}+1)}(\id_{AA'}+F_{AA'})\ ,
\end{align}
where $\{\ket{k}\}$ denotes some orthonormal basis of $\HC_{A}$.
\end{lemma}

Hence, our main result also holds for unitary two-designs.

\begin{corollary}
Let $\{U_\theta\}_{\theta\in\Theta}$ be a unitary two-design on some Hilbert space $\HC_A$. Then, we have for any bipartite quantum state $\rho_{AB}$ that
\begin{align}
H_{2}(K|B\Theta)=\log(d_{A}+1)-\log\left(2^{-H_2(A|B)}+1\right)\ ,
\end{align}
where
\begin{align}
\rho_{KB\Theta}=\frac{1}{|\Theta|}\sum_{\theta,k}(\dya{\theta_{k}}\ot\id_{B})\rho_{AB}(\dya{\theta_{k}}\ot\id_{B})\otimes\dya{\theta}_{\Theta}\ ,
\end{align}
and $\ket{\theta_{k}}=U_{\theta}\ket{k}$ for some orthonormal basis $\{\ket{k}\}$ of $\HC_{A}$.
\end{corollary}


\subsection{More general entropies}\label{app:more_general}

Our main result~\eqref{eq:mainResult} is in terms of the conditional R\'enyi 2-entropy. However, we can also prove it in terms of a more general continuous family of conditional 2-entropies. For $\nu\in[0,1]$ and bipartite quantum states $\rho_{AB}$, we define
\begin{align}\label{eq:nuentropies}
H_{2,\nu}(A|B)= -\log\Tr\left[\rho_{AB,\nu}\ad \rho_{AB,\nu}\right]\ ,
\end{align}
where
\begin{align}
\rho_{AB,\nu}=\big(\id_A\ot\rho_B^{-(1-\nu)/4}\big)\rho_{AB}\big(\id_A \ot \rho_B^{-(1+\nu)/4}\big)\ .
\end{align}
It is easily seen that $\nu=0$ corresponds to the usual definition~\eqref{eq:h2def} in the main text. We state the generalization of our uncertainty equality~\eqref{eq:mainResult} for a full set of MUBs, but note that it also holds for all other ``informationally equivalent'' measurements (cf.~Appendix~\ref{app:revisited}).

\begin{corollary}\label{thm:mainapp}
Let $\{\Theta\}_{\theta\in\Theta}$ be a complete set of MUBs on some Hilbert space $\HC_A$, and denote $\theta=\{\ket{\theta_{k}}\}_{k=1}^{d_{A}}$. Then, we have for any bipartite quantum state $\rho_{AB}$ that
\begin{align}\label{eqn25}
H_{2,\nu}(K|B\Theta)= \log(d_A + 1)-\log\left(2^{-H_{2,\nu}(A|B)} + 1\right)\ , 
\end{align}
where
\begin{align}
&\rho_{KB\Theta}=\notag\\
&\frac{1}{d_{A}+1}\sum_{\theta,k}(\dya{\theta_{k}}\ot\id_{B})\rho_{AB}(\dya{\theta_{k}}\ot\id_{B})\otimes\dya{\theta}_{\Theta}\ .
\end{align}
\end{corollary}

The proof is obvious by just taking the one for the conditional R\'enyi 2-entropy, and replacing $\tilde{\rho}_{AB}$ with $\rho_{AB,\nu}$. It is worth noting that~\eqref{eqn25} applies to the variant (used in e.g.~\cite{TomColRen09,Hayashi12}),
\begin{align}\label{eq:H2prime}
H_{2,1}(A|B)= -\log \Tr \left[ \rho_{AB}^2 (\id_A \ot \rho_B^{-1})\right]\ .
\end{align}


%


\end{document}